\documentclass[smallextended]{svjour3}       
\smartqed  
\usepackage{graphicx}
\usepackage{amsmath}
\usepackage{amssymb}
\usepackage{mathptmx}      
\usepackage[round]{natbib}

\journalname{Journal of Mathematical Biology}

\begin{document}

\title{Territorial pattern formation in the absence of an attractive potential
}

\titlerunning{Territory formation without attraction}

\author{Jonathan R. Potts         \and
        Mark A. Lewis
}


\institute{Jonathan R. Potts \at
              School of Mathematics and Statistics, University of Sheffield, UK \\
              Tel.: +44(0)114-222-3729\\
              \email{j.potts@sheffield.ac.uk}            \\
           \and
           Mark A. Lewis \at
              Centre for Mathematical Biology, Department of Mathematical and Statistical Sciences, University of Alberta, Canada
}

\date{Received: date / Accepted: date}

\maketitle

\begin{abstract} 

Territoriality is a phenomenon exhibited throughout nature.  On the individual level, it is the processes by which organisms exclude others of the same species from certain parts of space.  On the population level, it is the segregation of space into separate areas, each used by subsections of the population.  Proving mathematically that such individual-level processes can cause observed population-level patterns to form is necessary for linking these two levels of description in a non-speculative way.  Previous mathematical analysis has relied upon assuming animals are attracted to a central area.  This can either be a fixed geographical point, such as a den- or nest-site, or a region where they have previously visited.  However, recent simulation-based studies suggest that this attractive potential is not necessary for territorial pattern formation.  Here, we construct a partial differential equation (PDE) model of territorial interactions based on the individual-based model (IBM) from those simulation studies.  The resulting PDE does not rely on attraction to spatial locations, but purely on conspecific avoidance, mediated via scent-marking.  We show analytically that steady-state patterns can form, as long as (i) the scent does not decay faster than it takes the animal to traverse the terrain, and (ii) the spatial scale over which animals detect scent is incorporated into the PDE.  As part of the analysis, we develop a general method for taking the PDE limit of an IBM that avoids destroying any intrinsic spatial scale in the underlying behavioral decisions.  

\end{abstract}

\section{Introduction}
\label{intro}

Territoriality is a wide-spread phenomenon throughout nature.  A territory is an area of space used exclusively by an organism, or a group of organisms \citep{burt1943}.  It is formed by deliberately excluding others of the same species (called {\it conspecifics}) from the area, either by aggressive confrontations or mutual consent \citep{adams2001}.  In the last two decades, there have been a number of studies that show analytically how territorial patterns can form from the movements and interactions of animals \citep{lewismurray1993, moorcroftlewis2006, pottslewis2014}.  These use mean-field approximations to model the animals' behavioral decisions as partial differential equations (PDEs), and so enable territory formation to be analyzed using standard tools from PDE theory \citep{murray2002}.  

Despite their success in uncovering drivers behind space use patterns \citep{moorcroftetal2006}, previous analytical models assume an attractive potential influencing the animals' movements.  This could either be fidelity to a central place such as a den- or nest-site \citep{lewisetal1997}, or a tendency to move towards places that the animal has previously visited \citep{briscoeetal2002}.  However, it is not clear that such an attractive potential is in fact necessary for territory formation \citep{moorcroft2012}.

In this paper, we present a PDE model of territorial pattern formation based purely on conspecific avoidance, with no attractive potential.  It is based on an individual based model (IBM) of so-called {\it territorial random walkers}  \citep{GPH1}. 
Previous work used simulation analysis to demonstrate empirically that territories can form in this system \citep{GPH1}.  
Here, we show analytically the circumstances under which territorial patterns may form.  
Specifically, necessary conditions for territorial pattern formation include
\begin{itemize}
\item spatial aversion to scent marks
\item scent marks that persist for longer than it takes the animal to traverse the terrain, and
\item a reaction to conspecific scent averaged over a small area around the animal.
\end{itemize}

As is often the case in ecological applications, it is important that the discrete spatial nature of each interaction is present in the model \citep{durrettlevin1994}.  In the case of territorial interactions, this discreteness is inherent in the fact that animals have a non-zero perceptive radius for determining the presence of scent.  As part of this study, we develop a limiting procedure that enables the transition from IBM to PDE without losing this important aspect of spatially discrete interactions.  This has the potential for general use, as previous limiting procedures have often failed in this regard \citep{durrettlevin1994}.

The paper is organized as follows.  Section \ref{sec:IBM2PDE} derives the PDE from the IBM model.   Sections \ref{sec:tp} and \ref{sec:lin_al} investigate the conditions under which patterns may form. Section \ref{sec:conc} gives some concluding remarks.

\section{From the individual-level description to a system of PDEs}
\label{sec:IBM2PDE}

\subsection{Description of model}
\label{sec:model_desc}

The individual based model (IBM) is based on a 1D model of territoriality which was recently proposed by \citet{GPH1}, but then slightly modified and studied in detail by \citet{GPH2} and \citet{PHG2}.  The model consists of two agents moving on a 1D lattice.  The agents represent either a single individual responsible for territorial defense, or a group of individuals moving together, such as a pack or a flock.  For example, the former is appropriate when modelling fox ({\it Vulpes vulpes}) behaviour where the dominant male in each group marks and secures the territory \citep{harris1980}, whereas the latter may be more appropriate for modelling wolf ({\it Canis lupus}) packs \citep{lewismurray1993}.

Agents move as discrete-time discrete-space nearest-neighbour random walkers, depositing scent marks as they move.  In the model of \citet{GPH2}, the scent remains present for a finite amount of time, called the {\it active scent time} and denoted by the symbol $T_{\rm AS}$.  Once this time is up, provided the lattice site has not been re-scented, the mark is no longer considered by conspecifics to be `active'.  Agents cannot move into any lattice site that contains the active scent of another agent (Figure \ref{model_expl}a).

\begin{figure*}
	\includegraphics[width=\textwidth]{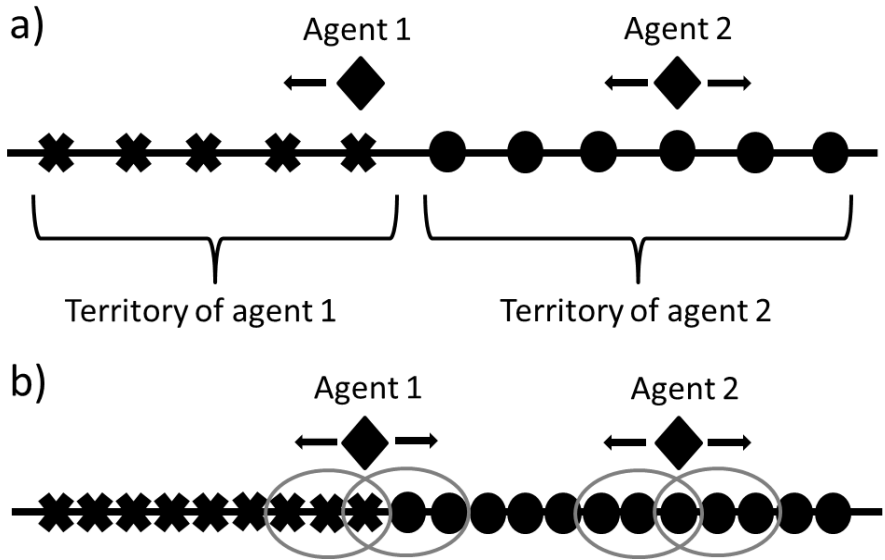}
	\caption{{\bf Pictorial representation of the underlying model.} The territory of each agent represents the sites containing that agent's scent.  In panel (a), agent 1 is unable to move to the right in the next step, since there is active scent of agent 2 there.  However, agent 2 can move in either direction.  In panel (b), we show the case where the lattice spacing, $a$, is halved, so response to scent is averaged over several sites, given by the grey ovals.  As the lattice spacing is reduced by a factor of $h(a)$, so the response to scent is averaged over $2h(a)-1$ sites.  In this situation, agent 1 has a higher probability of moving left than right, while agent 2 has equal probability of moving in both directions.}
	\label{model_expl}       
\end{figure*}

Our model set-up will take three stages.  Stage 1 uses the formalism of coupled step selection functions \citep{PML2013} to describe a stochastic IBM algorithm which generalizes that of \citet{GPH2}.  Stage 2 describes how to derive a mean-field probabilistic model from the IBM.  Stage 3 involves taking the PDE limit of the probabilistic model.

\begin{table}
\caption{{\bf Glossary of symbols.}  The first column shows the symbol, the second a definition, and the third whether it pertains to the discrete (lattice) model or the continuous limit or both.  Note that some symbols are used either as dimensional quantities or their dimensionless equivalents, depending on the context (see section 2.5).}
\begin{tabular}{l|l|l}
 Symbol  & Definition & Model \\
\hline
n & Arbitrary lattice site & Discrete \\
m & Arbitrary time step & Discrete \\
$E_i(n,m)$ & For animal $i$, the probability that there is conspecific scent at $(n,m)$ & Discrete \\
$\tau$ & Length of a single time step & Discrete \\
$a$ & Lattice spacing & Discrete \\
$\kappa(a)$ & Probability that scent is deposited when the agent visits a lattice site & Discrete \\
$h(a)$ & Number of lattice sites constituting the agent's perceptive radius & Discrete \\
$f_i^m(n|n')$ & Probability of agent $i$ moving to $n$ next jump, given it is at $n'$ at timestep $m$  & Discrete \\
$U(n,m)$ & Probability that agent 1 is at $n$ at timestep $m$  & Discrete\\
$V(n,m)$ & Probability that agent 2 is at $n$ at timestep $m$  & Discrete\\
$P(n,m)$ & Probability that scent of agent 1 is present at $n$ at timestep $m$  & Discrete\\
$Q(n,m)$ & Probability that scent of agent 2 is present at $n$ at timestep $m$  & Discrete \\
\hline
$\mu$ & Mean rate of scent decay & Both \\
$\lambda$ & Mean scent deposition over a unit of space in a unit of time & Both \\
\hline
x & Arbitrary position in continuous space & Continuous \\
t & Arbitrary continuous time & Continuous \\
$u(x,t)$ & Probability density function of agent 1 at time $t$ & Continuous \\
$v(x,t)$ & Probability density function of agent 2 at time $t$ & Continuous \\
$p(x,t)$ & Probability that scent of agent 1 is present at $(x,t)$ & Continuous \\
$q(x,t)$ & Probability that scent of agent 2 is present at $(x,t)$ & Continuous \\
$D$ & Diffusion constant & Continuous \\
$\delta$ & The agent's perceptive radius & Continuous \\
$\bar{p}(x,t)$ & Mean of $p(x,t)$ in a $\delta$-ball around $x$ & Continuous \\
$\bar{q}(x,t)$ & Mean of $q(x,t)$ in a $\delta$-ball around $x$ & Continuous \\
$L$ & Width of terrain & Continuous \\
$m$ & Dimensionless composite variable $\mu L/\lambda$ & Continuous \\
$\epsilon$ & Dimensionless composite variable $D/L\lambda$ & Continuous
\end{tabular}
\label{glossary}
\end{table}

\subsection{Stochastic algorithm for the individual based model}
\label{sec:IBM2PDE_IBM}

If unconstrained by scent marks, an agent is simply a nearest-neighbor random walker.  Therefore the probability that agent $i$ ($i \in \{1,2\}$) moves from site $n'$ to $n$ is $\phi_i(n|n')=1/2$ if $|n-n'|=1$ and $\phi_i(n|n')=0$ otherwise.  This function $\phi_i(n|n')$ is the {\it environment-independent movement kernel}.

Now we add the effect of scent marks, which for this paper are thought of as constituting of the animal's `environment'.  For each agent $i$, let the environment, ${E}_i(n,m)$, be the probability that there is conspecific scent at lattice site $n$ and timestep $m$.  We give two possible definitions for ${E}_i(n,m+1)$, denoted by ${E}^1_i(n,m+1)$ and ${E}^2_i(n,m+1)$, and both defined in terms of the state of the system at timestep $m$.  The first is given by
\begin{align}
{E}^1_i(n,m+1)=\begin{cases}1 & \text{if an agent $j\neq i$ is at position $n$ at any time }\\
&\text{between $m-T_{AS}+1$ and $m$,}
\\
0 & \text{otherwise.}
\end{cases}
\label{cssf_population1}
\end{align}
If ${E}^1_i(n,m)=1$ then there is conspecific scent present, otherwise there is not.  This is the definition used by \citet{GPH2} and \citet{PHG2}.  

An alternative to equation (\ref{cssf_population1}) is the following definition
\begin{align}
{E}^2_i(n,m+1)=\begin{cases}1-\mu\tau & \text{with probability $\kappa(a)$, if an agent $j\neq i$ is at}\\
&\text{position $n$ at timestep $m$,}
\\
(1-\mu\tau){E}^2_i(n,m) & \text{otherwise,}
\end{cases}
\label{cssf_population2}
\end{align}
where $\tau$ is the length of a timestep and $\kappa(a)$ is the probability that scent is deposited when the animal visits a lattice site.  Notice that scent left at timestep $m$ has a probability $1-\mu\tau$ of remaining present at timestep $m+1$.  

Introducing $\kappa(a)$ allows us to change the lattice spacing $a$ without changing the average distance moved between scent depositions, by insisting that $a/\kappa(a)$ is kept constant.  From Section \ref{sec:IBM2PDE_agent} onwards, we will use equation (\ref{cssf_population2}) to describe scent deposition and decay.  However, the stochastic algorithm of this section can be defined equally well using either equation (\ref{cssf_population1}) or (\ref{cssf_population2}).

We now define the interaction term, which denotes how the scent affects the agent's movement.  Animals typically have a perceptive radius that determines the spatial area over which they respond to scent.  The model of \citet{GPH2} implicitly identified this perceptive radius with the lattice spacing $a$.  However, this limits the model's flexibility: if the lattice spacing is changed then the model assumptions about the animal's perceptive radius are also changed.  Therefore, to ensure our model is not constrained by the choice of $a$, we define the interaction term, ${C}_i^j(n,m)$, to be a Bernoulli random variable taking value 1 with probability 
\begin{align}
\varphi=1-\frac{1}{2h(a)-1}\sum_{l=1-h(a)}^{h(a)-1}{E}^j_i(n+l,m),
\label{varphi}
\end{align}
where $j \in \{1,2\}$, $h(a)$ is defined so that $ah(a)$ is the perceptive radius of the animal, and ${C}_i^j(n,m)$ takes value 0 with probability $1-\varphi$.  The model from \citet{GPH2} implicitly had $h(a)=1$.  In general, to change the lattice spacing whilst keeping the perceptive radius $\delta$ constant requires setting $h(a)=\delta/a$, which holds as long as $h(a)$ is an integer (see Figure \ref{model_expl}b).

The probability $f_i^m(n|n')$ of agent $i$ moving to $n$ at timestep $m$, given that is was previously at position $n'$, is a combination of $\phi_i(n|n')$ and ${C}_i^j(n,m)$, written as follows
\begin{align}
f_i^m(n|n')=\begin{cases}\frac{\phi_i(n|n'){C}_i^j(n,m)}{\phi_i(n'+1|n'){C}_i^j(n'+1,m)+\phi_i(n'-1|n'){C}_i^j(n'-1,m)} & \text{if ${C}_i^j(n'+1,m)+ {C}_i^j(n'-1,m)\neq 0$},
\\
\delta_k(n-n') & \text{otherwise,}
\end{cases}
\label{cssf_jump_prob}
\end{align}
where $\delta_k$ is the Kronecker delta.

Equation (\ref{cssf_jump_prob}) allows us to describe the stochastic algorithm.  This is a one-step Markov process, so can be fully described by determining the possible states of the system at timestep $m+1$, given the state at time $m$.  Suppose that, for some $m$, we know ${E}^j_i(n,m)$ for every $n$.  Suppose further that animal $i$ is at position $n_i$ at timestep $m$.  Then the algorithm is as follows
\begin{enumerate}
\item Calculate ${E}^j_i(n,m+1)$ for each $n$.
\item Define a categorical distribution taking one of three values $n_i-1,n_i,n_i+1$ with probabilities given by $f_i^m(n_i-1|n_i)$, $f_i^m(n_i|n_i)$ and $f_i^m(n_i+1|n_i)$ respectively.  These values are the possible future positions of animal $i$.  
\item Draw a random variable from this categorical distribution and move the animal to the position just drawn.  
\item Repeat steps 2 and 3 for each animal in turn.
\end{enumerate}

\subsection{Probability distribution of an agent in a given scent distribution}
\label{sec:IBM2PDE_agent}

To construct a probabilistic master equation describing the above stochastic process, we first assume that the evolution of the scent marks can be decoupled from the movement of the agent.  In other words, we calculate the equation governing a single step of each agent's movement that is true for any fixed, arbitrary scent distribution of the other agent.  This is a so-called mean-field approximation, that assumes covariates between the agent and conspecific scent are small enough to ignore.  

Let $U(n,m)$ (resp. $V(n,m)$) be the probability of agent $1$ (resp. $2$) being at position $n$ at timestep $m$ and $P(n,m)$ (resp. $Q(n,m)$) the probability of there being scent present of agent $1$ (resp. $2$) at position $n$ at timestep $m$.  By analysing the probability of moving to site $n$ from either site $n-1$, $n$, or $n+1$ in one timestep, we eventually arrive at the following discrete space-time master equations
\begin{align}
U(n,m+1)=&[1-Q(n+i,m)]\biggl\{ \frac{1}{2}U(n-1,m)[1+Q(n+i-2,m)]+ \nonumber \\
&\frac{1}{2}U(n+1,m)[1+Q(n+i+2,m)]\biggr\} + \nonumber \\
&U(n,m)Q(n+i-1,m)Q(n+i+1,m), \label{prob_me1} \\
V(n,m+1)=&[1-P(n+i,m)]\biggl\{ \frac{1}{2}V(n-1,m)[1+P(n+i-2,m)]+\nonumber \\
&\frac{1}{2}V(n+1,m)[1+P(n+i+2,m)]\biggr\} + \nonumber \\
&V(n,m)P(n+i-1,m)P(n+i+1,m),
\label{prob_me2}
\end{align}
where the following implicit summation notation \citep{einstein1916} is used
\begin{align}
\label{sum_not}
P(n+i,t):=\frac{1}{2h(a)-1}\sum_{i=1-h(a)}^{h(a)-1} P(n+i,t), \nonumber \\
Q(n+i,t):=\frac{1}{2h(a)-1}\sum_{i=1-h(a)}^{h(a)-1} Q(n+i,t),
\end{align}
and $a$ is the lattice spacing and the product $ah(a)$ is the perceptive radius of the agent.  

To give some intuition behind equations (\ref{prob_me1}) and (\ref{prob_me2}), we focus on equation (\ref{prob_me1}), and note that all of the comments in this paragraph hold equally well for equation (\ref{prob_me2}).  The initial $[1-Q(n+i,m)]$ factor in equation (\ref{prob_me1}) ensures that there is a low probability of moving to position $n$ if there is a high probability of active conspecific scent being present at or around position $n$.  The factor $[1+Q(n+i-2,m)]$ (resp. $[1+Q(n+i+2,m)]$) means that if scent is likely to be present at or around $n-2$ ($n+2$) and the animal is at $n-1$ ($n+1$) at time $m$ then it will be likely to move to $n$ at time $m+1$.  The final summand $U(n,m)Q(n+i-1,m)Q(n+i+1,m)$ means that if the presence of scent is highly probable both to the left and right of an animal at time $m$, then it is likely to stay where it is.  Notice that if $\sum_n U(n,m)=1$ then $\sum_n U(n,m+1)=1$ so that probabilities are conserved. 

Let $\tau$ be the waiting-time between successive jumps.  Then equations (\ref{prob_me1}) and (\ref{prob_me2}) rearrange to give
\begin{align}
\frac{U(n,m+1)-U(n,m)}{\tau}=&\frac{1}{2\tau}[1-Q(n+i,m)]\{U(n-1,m)[1+Q(n+i-2,m)]+\nonumber\\ 
&U(n+1,m)[1+Q(n+i+2,m)]\}-\nonumber\\ 
&\frac{1}{\tau}[1-Q(n+i-1,m)Q(n+i+1,m)]U(n,m),\label{disc_scent_ave1} \\ 
\frac{V(n,m+1)-V(n,m)}{\tau}=&\frac{1}{2\tau}[1-P(n+i,m)]\{V(n-1,m)[1+P(n+i-2,m)]+\nonumber \\
&V(n+1,m)[1+P(n+i+2,m)]\}-\nonumber\\ 
&\frac{1}{\tau}[1-P(n+i-1,m)P(n+i+1,m)]V(n,m),\label{disc_scent_ave2}
\end{align}
Equation (\ref{disc_scent_ave1}) can be re-written as follows
\begin{align}
\label{ume2}
&\frac{U(n,m+1)-U(n,m)}{\tau}=\frac{a^2}{2\tau}\biggl\{\frac{1}{a}\left[\frac{U(n+1,t)-U(n,t)}{a}-\frac{U(n,t)-U(n-1,t)}{a}\right]+\\ \nonumber
&\qquad\frac{1}{2a}\biggl[4U(n+1,t)\frac{Q(n+i+2,t)-Q(n+i,t)}{2a}-\\ \nonumber
&\qquad 4U(n-1,t)\frac{Q(n+i,t)-Q(n+i-2,t)}{2a}\biggr] +\\ \nonumber
&\qquad \frac{1}{a}\biggl[\frac{U(n,t)Q(n+i+1,t)Q(n+i-1,t)-U(n-1,t)Q(n+i,t)Q(n+i-2,t)}{a}-\\ \nonumber
&\qquad \frac{U(n+1,t)Q(n+i+2,t)Q(n+i,t)-U(n,t)Q(n+i+1,t)Q(n+i-1,t)}{a}\biggr]\biggr\},
\end{align}
and similarly for equation (\ref{disc_scent_ave2}).  Taking the limit as $a,\tau \rightarrow 0$ and $n,m,h(a)\rightarrow \infty$ such that $D=a^2/(2\tau)$, $x=na$, $t=m\tau$, $ah(a)=\delta$ in the limit, and writing $u(x,t)$ (resp. $v(x,t)$) for the probability density functions of agent 1's (resp. 2's) position and $p(x,t)$ (resp. $q(x,t)$) for the probability that agent 1's (resp. 2's) active scent is present at position $x$ at time $t$, we arrive at the following PDE (see Appendix A for a full derivation)
\begin{align}
\label{qu_time}
\frac{\partial u}{\partial t}=D\frac{\partial^2}{\partial x^2}[(1-\bar{q}^2)u]+4D\frac{\partial}{\partial x}\left[\frac{\partial \bar{q}}{\partial x}u\right].
\end{align}
The equation governing the evolution of $v(x,t)$ over time is analogous
\begin{align}
\label{pv_time}
\frac{\partial v}{\partial t}=D\frac{\partial^2}{\partial x^2}[(1-\bar{p}^2)v]+4D\frac{\partial}{\partial x}\left[\frac{\partial \bar{p}}{\partial x}v\right].
\end{align}
Here, $\bar{p}(x,t)$ and $\bar{q}(x,t)$ are the locally averaged scent of agents 1 and 2, respectively
\begin{align}
\label{barp}
\bar{p}(x,t)\frac{1}{2\delta}\int_{-\delta}^{\delta}p(x+z,t){\rm d}z,
\end{align}
\begin{align}
\label{barq}
\bar{q}(x,t)\frac{1}{2\delta}\int_{-\delta}^{\delta}q(x+z,t){\rm d}z. 
\end{align}

\subsection{Evolution of the scent distribution}
\label{sec:IBM2PDE_scent}

Recall that we gave two different formulae for the scent decay process, equations (\ref{cssf_population1}) and (\ref{cssf_population2}).  For the purposes of our mean-field analysis, it is convenient to use equation (\ref{cssf_population2}).  In other words, the probability of scent being present at lattice site $n$ decays by a factor of $1-\tau\mu$ each timestep of length $\tau$.  Additionally, when a site is visited by the animal, the probability that there is active scent present jumps to 1 with probability $\kappa(a)$.  

The following master equation follows directly from taking the expectation of either side of equation (\ref{cssf_population2}) 
\begin{align}
P(n,m+1)=(1-\mu\tau) U(n,m)\kappa(a) + (1-\mu\tau)\left[1- U(n,m)\kappa(a)\right]P(n,m).
\label{p_me_disc}
\end{align}
The probability density version of equation (\ref{p_me_disc}) is the limit as $a,\kappa(a),\tau\rightarrow 0$ and $n,m\rightarrow \infty$ of
\begin{align}
p(na,m\tau+\tau)=&(1-\mu\tau) u(na,m\tau)a\kappa(a) + (1-\mu\tau)\left[1- u(na,m\tau)a\kappa(a)\right]p(na,m\tau),
\label{p_me_cts2}
\end{align}
such that $x=na$, $t=m\tau$ and $\lambda=a\kappa(a)/\tau$ in this limit.  

Subtracting $p(na,m\tau)$ from both sides of equation (\ref{p_me_cts2}), dividing by $\tau$ and taking this limit leads to the following ordinary differential equation (ODE) governing $p(x,t)$ 
\begin{align}
\frac{\partial p}{\partial t}= \lambda (1-p)u - \mu p.
\label{pu_time}
\end{align}
We can interpret $\lambda$ as representing the amount of scent deposited over a unit of space in a single unit of time. The derivation for $q(x,t)$ is similar and gives
\begin{align}
\frac{\partial q}{\partial t}= \lambda (1-q)v - \mu q.
\label{qv_time}
\end{align}

Analyzing the system of equations (\ref{qu_time}), (\ref{pv_time}), (\ref{pu_time}) and (\ref{qv_time}) requires choosing an appropriate domain and boundary conditions.  A simple and biologically realistic choice is to assume that agents are confined in a domain $[0,L]$ with zero flux boundary conditions.  The boundary conditions could either come about by being confined in a valley or on a small island.  Alternatively, the conditions could model a situation where the rate of migration of animals into the domain is equal to the rate of movement outwards.  In other words, the population is assumed to be exhibiting a certain spatial and temporal stability.  These boundary conditions are given as follows
\begin{align}
\left\{\frac{\partial}{\partial x}[(1-\bar{q}^2)u]+4\left[\frac{\partial \bar{q}}{\partial x}u\right]\right\}\bigg|_{x=0}=\left\{\frac{\partial}{\partial x}[(1-\bar{q}^2)u]+4\left[\frac{\partial \bar{q}}{\partial x}u\right]\right\}\bigg|_{x=L}=0,
\label{qu_bdry}
\end{align}
\begin{align}
\left\{\frac{\partial}{\partial x}[(1-\bar{p}^2)v]+4\left[\frac{\partial \bar{p}}{\partial x}v\right]\right\}\bigg|_{x=0}=\left\{\frac{\partial}{\partial x}[(1-\bar{p}^2)v]+4\left[\frac{\partial \bar{p}}{\partial x}v\right]\right\}\bigg|_{x=L}=0.
\label{pv_bdry}
\end{align}

The existence of the boundary requires that we need to redefine $\bar{p}(x,t)$ and $\bar{q}(x,t)$ in the cases where $x<\delta$ and $x>L-\delta$, as follows
\begin{align}
\label{barp_redef}
\bar{p}(x,t)=\begin{cases}
\frac{1}{x+\delta}\int_{-x}^{\delta}p(x+z,t){\rm d}z & \mbox{if $x<\delta$}, \\
\frac{1}{2\delta}\int_{-\delta}^{\delta}p(x+z,t){\rm d}z & \mbox{if $\delta\leq x\leq L-\delta$}, \\
\frac{1}{L-x+\delta}\int_{-\delta}^{L-x}p(x+z,t){\rm d}z & \mbox{if $x>L-\delta$}, 
\end{cases}
\end{align}
\begin{align}
\label{barq_redef}
\bar{q}(x,t)=\begin{cases}
\frac{1}{x+\delta}\int_{-x}^{\delta}q(x+z,t){\rm d}z & \mbox{if $x<\delta$}, \\
\frac{1}{2\delta}\int_{-\delta}^{\delta}q(x+z,t){\rm d}z & \mbox{if $\delta\leq x\leq L-\delta$}, \\
\frac{1}{L-x+\delta}\int_{-\delta}^{L-x}q(x+z,t){\rm d}z & \mbox{if $x>L-\delta$}. 
\end{cases}
\end{align}

In addition to the boundary conditions, it is necessary to impose integral conditions on the initial probability distributions $u(x,0)$ and $v(x,0)$, to ensure that probability is conserved.  In other words
\begin{align}
\int_0^L u(x,0){\rm d}x=\int_0^L v(x,0){\rm d}x = 1.
\label{integral_conds_init}
\end{align}
A consequence of equations (\ref{qu_bdry}) and (\ref{pv_bdry}) is that the time-derivative of $\int_0^L u(x,t){\rm d}x$ is zero.  Therefore the initial conditions from equation (\ref{integral_conds_init}) imply that probabilities are conserved at every point in time, i.e.
\begin{align}
\int_0^L u(x,t){\rm d}x=\int_0^L v(x,t){\rm d}x = 1.
\label{integral_conds}
\end{align}

\subsection{A dimensionless version of the model}
\label{sec:IBM2PDE_dimless}

To minimize the number of model parameters, we re-write equations (\ref{qu_time}), (\ref{pv_time}), (\ref{pu_time}), and (\ref{qv_time}), using the following dimensionless parameters
\begin{align}
\tilde{u}=Lu,\,\tilde{v}=Lv,\,\tilde{x}=\frac{x}{L},\,\tilde{t}=\frac{tD}{L^2},\,m=\frac{\mu L}{\lambda},\,\epsilon=\frac{D}{L\lambda}.
\label{dimensionless_params}
\end{align}
Dropping the tildes over the letters to ease notation, we arrive at the following dimensionless system of equations, which will be the object of study for the rest of this paper
\begin{align}
\frac{\partial u}{\partial t}&=\frac{\partial^2}{\partial x^2}\left[(1-\bar{q}^2)u \right]+4\frac{\partial}{\partial x}\left[\frac{\partial \bar{q}}{\partial x}u\right],
\label{qu_dimless}
\end{align}
\begin{align}
\frac{\partial v}{\partial t}&=\frac{\partial^2}{\partial x^2}\left[(1-\bar{p}^2)v \right]+4\frac{\partial}{\partial x}\left[\frac{\partial \bar{p}}{\partial x}v\right],
\label{pv_dimless}
\end{align}
\begin{align}
\epsilon\frac{\partial p}{\partial t}= (1-p)u - mp,
\label{pu_dimless}
\end{align}
\begin{align}
\epsilon\frac{\partial q}{\partial t}= (1-q)v - mq.
\label{qv_dimless}
\end{align}

\section{Territorial patterns}
\label{sec:tp}

We define a {\it territorial pattern} to be a non-trivial steady-state solution to equations (\ref{qu_dimless})-(\ref{qv_dimless}).  These are found by setting to zero the left-hand sides of equations (\ref{qu_dimless})-(\ref{qv_dimless}).  Setting equation (\ref{pu_dimless}) (resp. equation \ref{qv_dimless}) to zero enables the steady state solution of $p(x,t)$ (resp. $q(x,t)$), denoted by $p^\ast(x)$ (resp. $q^\ast(x)$), to be written in terms of the steady state solution of $u(x,t)$ (resp. $v(x,t)$), denoted by $u^\ast(x)$ (resp. $v^\ast(x)$) as follows
\begin{align}
p^\ast(x)= \frac{u^\ast(x)}{m + u^\ast(x)},
\label{pu_ast}
\end{align}
\begin{align}
q^\ast(x)= \frac{v^\ast(x)}{m + v^\ast(x)}.
\label{qv_ast}
\end{align}
To ease notation, we will henceforth drop the asterisks.  By setting equations (\ref{qu_dimless}) and (\ref{pv_dimless}) to zero and integrating with respect to $x$, we have that
\begin{align}
\frac{\rm d}{{\rm d} x}\{(1-\bar{q}[v(\cdot),x]^2)u(x)\}+4\left[\frac{{\rm d} \bar{q}}{{\rm d} x}u(x)\right] = c_1,
\label{qu_ss}
\end{align}
\begin{align}
\frac{\rm d}{{\rm d} x}\{(1-\bar{p}[u(\cdot),x]^2)v(x)\}+4\left[\frac{{\rm d} \bar{p}}{{\rm d} x}v(x)\right] = c_2,
\label{pv_ss}
\end{align}
for constants $c_1$ and $c_2$.  The boundary conditions given by equations (\ref{qu_bdry}) and (\ref{pv_bdry}) imply that $c_1=c_2=0$.  

We use the notation $\bar{p}[u(\cdot),x]$ and $\bar{q}[v(\cdot),x]$ to emphasize the fact that $\bar{p}$ and $\bar{q}$ are functionals.  That is, they map the functions $u(\cdot)$ and $v(\cdot)$, respectively, to the interval $[0,1]$.  These functionals are given by the following formulae
\begin{align}
\label{p_bar_ss}
\bar{p}[u(\cdot),x]=\begin{cases}
\frac{1}{x+\delta}\int_{-x}^{\delta} \frac{u(x+z)}{m+u(x+z)} {\rm d}z, & \mbox{if $x<\delta$}, \\
\frac{1}{2\delta}\int_{-\delta}^{\delta} \frac{u(x+z)}{m+u(x+z)} {\rm d}z & \mbox{if $\delta\leq x\leq 1-\delta$}, \\
\frac{1}{1-x+\delta}\int_{-\delta}^{1-x}\frac{u(x+z)}{m+u(x+z)}{\rm d}z & \mbox{if $x>1-\delta$},
\end{cases} \\
\bar{q}[v(\cdot),x]=\begin{cases}
\frac{1}{x+\delta}\int_{-x}^{\delta} \frac{v(x+z)}{m+v(x+z)} {\rm d}z, & \mbox{if $x<\delta$}, \\
\frac{1}{2\delta}\int_{-\delta}^{\delta} \frac{v(x+z)}{m+v(x+z)} {\rm d}z & \mbox{if $\delta\leq x\leq 1-\delta$}, \\
\frac{1}{1-x+\delta}\int_{-\delta}^{1-x}\frac{v(x+z)}{m+v(x+z)}{\rm d}z & \mbox{if $x>1-\delta$},
\end{cases}
\label{q_bar_ss}
\end{align}

In sum, as well as equations (\ref{p_bar_ss}) and (\ref{q_bar_ss}), we have the following system of equations, whose non-constant solutions correspond to territorial patterns
\begin{align}
\frac{\rm d}{{\rm d} x}\{(1-\bar{q}[v(\cdot),x]^2)u(x)\}+4\left[\frac{{\rm d} \bar{q}}{{\rm d} x}u(x)\right] = 0,
\label{qu_ss2}
\end{align}
\begin{align}
\frac{\rm d}{{\rm d} x}\{(1-\bar{p}[u(\cdot),x]^2)v(x)\}+4\left[\frac{{\rm d} \bar{p}}{{\rm d} x}v(x)\right] = 0,
\label{pv_ss2}
\end{align}
\begin{align}
{p(x)}= \frac{u(x)}{m + u(x)},
\label{pu_ast_loc}
\end{align}
\begin{align}
{q(x)}= \frac{v(x)}{m + v(x)}.
\label{qv_ast_loc}
\end{align}

\subsection{Territorial patterns with only local interactions}

We first examine the case where $\delta\rightarrow 0$ so that agents only respond to scent at the exact position where they are situated.  This means equations (\ref{qu_ss2}) and (\ref{pv_ss2}) become
\begin{align}
\frac{\rm d}{{\rm d} x}\{(1-{q(x)}^2)u(x)\}+4\left[\frac{{\rm d} q}{{\rm d} x}u(x)\right] = 0,
\label{qu_ss_loc}
\end{align}
\begin{align}
\frac{\rm d}{{\rm d} x}\{(1-{p(x)}^2)v\}+4\left[\frac{{\rm d} p}{{\rm d} x}v(x)\right] = 0.
\label{pv_ss_loc}
\end{align}
The limit $\delta\rightarrow 0$ means that the functionals $\bar{p}[u(\cdot),x]$ and $\bar{q}[v(\cdot),x]$ have been replaced by functions $p(x)$ and $q(x)$, which makes analysis tractable.  To ease notation, we hencefore drop the explicit dependence of the functions $u$, $v$, $p$, and $q$ on $x$.

By substituting equations (\ref{pu_ast_loc}) and (\ref{qv_ast_loc}) into (\ref{qu_ss_loc}) and (\ref{pv_ss_loc}), the following system of ODEs for the steady state solution of $(u,v)$ is found
\begin{align}
A\dot{\bf u}&=0, \nonumber \\
A&=m\left( \begin{array}{cc}
(m+2v)(m+v) & 2u(2m+v) \\
2v(2m +u) & (m +2u)(m + u) \end{array} \right), \nonumber \\
\dot{\bf u} &= \left( \begin{array}{c}
{\rm d}u/{\rm d}x \\
{\rm d}v/{\rm d}x \end{array} \right).
\label{uv_ss_a}
\end{align}
The system of ODEs in equation (\ref{uv_ss_a}) is simple enough to analyze mathematically.  The results of this analysis are summarized in the following 
\begin{theorem}
\label{NoTerrsLocInt}
\begin{enumerate}
\item {\bf No scent decay.} If $m=0$ then $p(x)=q(x)=1$ and $u(x)$, $v(x)$ can take any value.
\item {\bf Positive scent decay.} If $m>0$ then there are no non-constant solutions to equation (\ref{uv_ss_a}).  Hence no territorial patterns can form in this case.
\end{enumerate}
\end{theorem}
\begin{proof}
See appendix B.
\end{proof}

\subsection{Territorial patterns with non-local interactions}

In the case where $\delta>0$, equations (\ref{qu_ss2})-(\ref{qv_ast_loc}) give a system of integral-ODEs, so are harder to analyse analytically.  Instead, we solve them numerically using the method of false transients \citep{mallinsondevahldavis1973}.  This involves solving equations (\ref{qu_dimless})-(\ref{qv_dimless}) forward in time until the solution is unchanging.  

Our algorithm uses a forward-difference approximation for time and a central difference approximation for space.  We divide the interval $[0,1]$ into $1,000$ equal, non-intersecting, sub-intervals of length $0.001$.  We iterate finite-difference versions of equations (\ref{qu_dimless})-(\ref{qv_dimless}) using timesteps of $0.01$, until all of the $u(x,t)$ or $v(x,t)$ values in all of the sub-sections are increasing by less than $10^{-8}$ over each timestep.  The initial conditions have all of $u(x)$ concentrated on the sub-interval $[0.25,0.251)$ and all of $v(x)$ on the sub-interval $[0.75,0.751)$.  This means $u(x)$ and $v(x)$ are zero outside the sub-intervals $[0.25,0.251)$ and $[0.75,0.751)$ respectively, and each integrate to 1 over $[0,1]$. 

Numerical analysis shows that patterns emerge from this system corresponding to two territories: $u(x)$ on the left and $v(x)$ on the right (figure \ref{pde_numerics}a,b).  Notice that a larger scent averaging radius leads to wider overlap of the probability distributions, meaning that the perceptive scale of the animal plays a large role in the territorial patterns that emerge.  

These can be compared with the territories that form in the original IBM with the interaction rules from \citet{GPH2}.  Although there is some qualitative agreement, the patterns generated by the IBM are still quite different to the PDE.  In the IBM, at any point in time, there is a border between the two territories.  This border fluctuates about the central point, typically much slower than the movement of the agent.  Each agent is free to move within its territory borders.  Consequently, the probability density of both agents combined ($u+v$) ends up being roughly uniform (figure~\ref{pde_numerics}c).  This does not happen in the mean field approximation studied here.  Indeed, the value of $u+v$ appears to be lower in the middle of the terrain.  Since this is just an artifact of the assumptions made in using the PDE limit, it is necessary to be cautious when inferring biological lessons from such pattern features.

\begin{figure*}
\includegraphics[width=\textwidth]{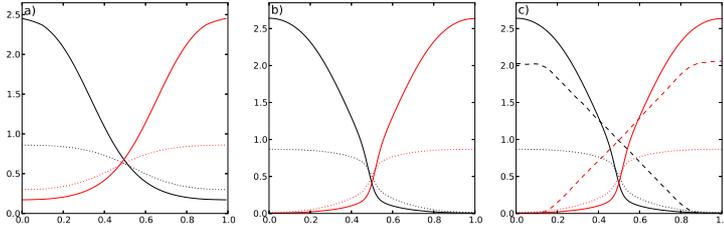}
\caption{{\bf Numerical steady state solutions of the model.} Solid red (resp. black) lines denote values of $u(x)$ (resp. $v(x)$), whereas dotted red (resp. black) lines show values of $p(x)$ (resp. $q(x)$).  In both panels, $m=0.4$ and $\epsilon=0.01$.  In panel (a), $\delta=0.1$, whilst panel (b) has $\delta=0.01$.  Notice that a larger the scent-averaging radius, $\delta$, gives a larger overlap between territories.  Panel (c) compares steady states of simulations of the original IBM (dashed lines) with numerical results from the PDE approximation (solid lines).  As in panels (a) and (b), dotted red (resp. black) lines show values of $p(x)$ (resp. $q(x)$).  Here, $\delta=0.01$, $m=0.4$, and $\epsilon=0.01$.  This corresponds, in the IBM, to $T_{\rm AS}/\tau=500$ and $N=100$.
}
\label{pde_numerics}       
\end{figure*}

\section{Investigating pattern formation via linear analysis}
\label{sec:lin_al}

A common technique for examining whether patterns spontaneously form in a dynamical system is to linearize the system about the uniform steady state and examine the resulting dispersion relation, e.g. \citet{murray2002} chapter 2.  For our system, the uniform steady state is
\begin{align}
(u_s,v_s,p_s,q_s)=\left(1,1,\frac{1}{1+m},\frac{1}{1+m}\right).
\label{steady_state}
\end{align}
That $u_s=v_s=1$ arises from the integral conditions (equation \ref{integral_conds}).  The values for $p_s$ and $q_s$ then follow from equations (\ref{pu_ast_loc}) and (\ref{qv_ast_loc}).

\begin{figure*}
\includegraphics[width=\textwidth]{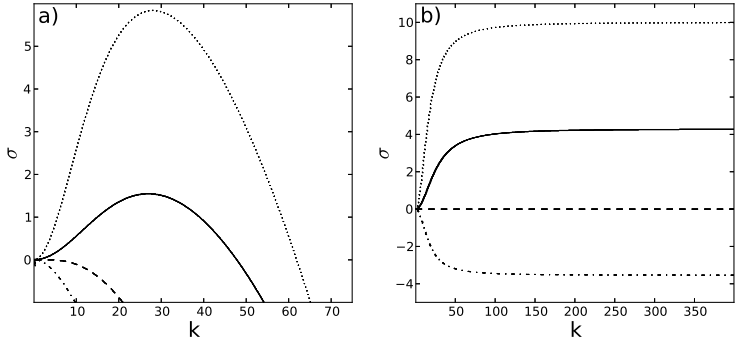}
\caption{{\bf Dispersion relations.}  Panel (a) show the dispersion relation for the dynamical system in equations (\ref{qu_dimless}-\ref{qv_dimless}), where $m=0.1$ (solid line), $m=0.5$ (dotted), $m=1$ (dashed), and $m=3$ (dot-dashed).  We set $\epsilon=0.01$ and $\delta=0.01$ throughout.  Panel (b) shows the same dispersion relations, but this time the animals respond only to the scent density at the particular point in which they reside, i.e. we take the limit $\delta\rightarrow 0$.  This system is given in equations (\ref{pu_dimless}), (\ref{qv_dimless}), (\ref{qu_noave}) and (\ref{pv_noave}).  The values of $\epsilon$ and $m$ are identical to those in panel (a).}
\label{disp_rel}       
\end{figure*}

Letting ${\bf w}=(\hat{u},\hat{v},\hat{p},\hat{q})=(u-u_s,v-v_s,p-p_s,q-q_s)$, we use equations (\ref{qu_dimless}-\ref{qv_dimless}) to give the linearized system
\begin{align}
\frac{\partial \hat{u}}{\partial t} &= \left[1-\frac{1}{(1+m)^2}\right]\frac{\partial^2\hat{u}}{\partial x^2}+2\left[2-\frac{1}{1+m}\right]\frac{\partial^2 \hat{\bar{q}}}{\partial x^2}, \nonumber \\
\frac{\partial \hat{v}}{\partial t} &= \left[1-\frac{1}{(1+m)^2}\right]\frac{\partial^2\hat{v}}{\partial x^2}+2\left[2-\frac{1}{1+m}\right]\frac{\partial^2 \hat{\bar{p}}}{\partial x^2}, \nonumber \\
\frac{\partial \hat{p}}{\partial t} &= \frac{m}{\epsilon(1+m)}\hat{u}-\frac{1+m}{\epsilon}\hat{p}, \nonumber \\
\frac{\partial \hat{q}}{\partial t} &= \frac{m}{\epsilon(1+m)}\hat{v}-\frac{1+m}{\epsilon}\hat{q}.
\label{linear_system}
\end{align}
Searching for solutions of the form $w=(u_0,v_0,p_0,q_0)\exp(\sigma t+{\rm i}kx)$, we obtain the following eigenvector equation
\begin{align}
A{\bf w}&=\sigma{\bf w} \nonumber \\
A&=\left(\begin{array}{cccc}
\left[\frac{1}{(1+m)^2}-1\right]k^2 & 0 & 0 & -2\left[2-\frac{1}{1+m}\right]\frac{k}{\delta}\sin{\delta k} \\
0 & \left[\frac{1}{(1+m)^2}-1\right]k^2 & -2\left[2-\frac{1}{1+m}\right]\frac{k}{\delta}\sin{\delta k} & 0 \\
\frac{m}{\epsilon(1+m)} & 0 & -\frac{1+m}{\epsilon} & 0 \\
0 & \frac{m}{\epsilon(1+m)} & 0 & -\frac{1+m}{\epsilon}
\end{array}\right).
\label{eigenvector_eqn2}
\end{align}
The dispersion relation is given by plotting the real values of $\sigma$ as a function of the wave number $k$, wherever $\mbox{det}(A-\sigma I)=0$.  As shown in Figure \ref{disp_rel}a, patterns can form for a finite range of wavelengths as long as $m<1$; that is, as long as the scent decay is not too rapid.  

We can gain biological insight by relating this result back to the underlying IBM.  Recall that $m=\mu L/\lambda$ (equation \ref{dimensionless_params}).  Recall also that $\lambda$ is the limit of $a\kappa(a)/\tau$.  In the original lattice model, where $\kappa(a)=1$, $a\kappa(a)/\tau$ is simply the speed of the animal. Then $m<1$ if and only if the time it would take a freely moving animal on the lattice to traverse the whole terrain is less than the characteristic timescale for scent-mark decay $1/\mu$. 

The dispersion relation changes somewhat if we examine the case where $\delta \rightarrow 0$, so that animals only respond to scent in the exact place that they are located at any point in time.  In this case, $\bar{p}$ and $\bar{q}$ are replaced by $p$ and $q$ respectively, so that equations (\ref{qu_dimless}) and (\ref{pv_dimless}) are replaced by
\begin{align}
\frac{\partial u}{\partial t}&=\frac{\partial^2}{\partial x^2}\left[(1-{q}^2)u \right]+4\frac{\partial}{\partial x}\left[\frac{\partial {q}}{\partial x}u\right],
\label{qu_noave}
\end{align}
\begin{align}
\frac{\partial v}{\partial t}&=\frac{\partial^2}{\partial x^2}\left[(1-{p}^2)v \right]+4\frac{\partial}{\partial x}\left[\frac{\partial {p}}{\partial x}v\right].
\label{pv_noave}
\end{align}
Equation (\ref{eigenvector_eqn2}) becomes
\begin{align}
A&=\left(\begin{array}{cccc}
\left[\frac{1}{(1+m)^2}-1\right]k^2 & 0 & 0 & -2\left[2-\frac{1}{1+m}\right]k^2 \\
0 & \left[\frac{1}{(1+m)^2}-1\right]k^2 & -2\left[2-\frac{1}{1+m}\right]k^2 & 0 \\
\frac{m}{\epsilon(1+m)} & 0 & -\frac{1+m}{\epsilon} & 0 \\
0 & \frac{m}{\epsilon(1+m)} & 0 & -\frac{1+m}{\epsilon}
\end{array}\right),
\label{eigenvector_eqn}
\end{align}
which is the limit as $\delta\rightarrow 0$ of equation (\ref{eigenvector_eqn2}).  The corresponding dispersion relation is given in Figure \ref{disp_rel}b.  Here, for $0<m<1$, $\sigma$ is an increasing function of $k$, indicating that the steady state is unstable but arbitrarily large wave numbers grow fastest.  In other words, this is an ill-posed problem.  

\section{Discussion and conclusions}
\label{sec:conc}

We have shown how stable territorial patterns can form purely from a conspecific avoidance mechanism, without requiring any attractive potential.  Our model is constructed by taking the continuous space-time limit of a discrete lattice model.  Therefore it can be rigorously linked to the underlying movement and interaction processes.  We have demonstrated that patterns will only form if the scent marks last for a sufficiently long time.  If they decay too quickly, i.e. $m\geq 1$, the territorial structure breaks down.  This can be interpreted as saying territories can only emerge if the animal is able to patrol its territory faster than the scent marks decay.

Similarly, patterns will only form reliably if the animals react to the averaged scent density across the local vicinity of the animal.  From a biological perspective, an animal will always have a perceptive radius over which it will react to scent.  Therefore this spatial averaging is implicit in the system being modeled.  As such, our study demonstrates the importance of ensuring that the mathematical limiting process, moving from discrete to continuous space, does not destroy a key feature of the underlying biology.  Our procedure for performing this limiting process has the potential for broad application, since there are many examples where the discreteness of ecological interactions is known to be an important feature of the modeling process \citep{durrettlevin1994}.

The model is derived from an individual-based model, previously studied using stochastic simulations \citep{GPH1, PHG2}.  As noted in recent reviews \citep{giuggiolikenkre2014, pottslewis2014}, one of the advantages of this approach is that it gives a clear delineation between the related notions of `home range' and `territory'.  The territory of an animal is defined as the area containing active scent marks of the animal \citep{burt1943}.  Therefore, in the model presented here, $p(x,t)$ and $q(x,t)$ can be considered the probabilities of position $x$ being part of the animals' {\it territories} at time $t$.

On the other hand, the home range of an animal is its utilization distribution \citep{burt1943}.  Therefore $u(x,t)$ and $v(x,t)$ can be considered as the {\it home ranges} of the animals at time $t$.  The utilization distribution of an animal is typically much easier to measure in the field than the fluctuating locations of the territory border \citep{PHG2}.  In our approach, the concepts of territory and home range are related by rather simple steady-state equations (\ref{pu_ast}) and (\ref{qv_ast}).  This gives an explicit way to calculate the probable location of a territory border, given data on its home range.

A key reason for studying PDE limits of IBMs is to provide mathematical analysis of the conditions under which patterns may form, rather than relying on empirical evidence from computer simulations.  However, as shown here, patterns that form from numerically solving the PDE may to be quantitatively different from those formed by simulating the IBM.  Therefore, if such PDE models were fitted to data on real systems, it is important for the user to check that the PDE results are not significantly different to those given by the IBM.  Otherwise, there is a danger of making incorrect inferences about biological patterns, that may merely arise as artifacts of the mean-field approximation and/or limiting procedure.

Models such as ours could be of use in analyzing territory formation when there is no reason to believe the animals have any fidelity towards particular locations, or where these locations are not known, e.g. \citet{batemanetal2015}.  Though memory processes have recently been invoked to explain pattern formation \citep{briscoeetal2002, moorcroft2012}, it is unclear how to find out what is going on inside the minds of the animals using current science.  This makes conjectures about memory difficult to falsify.  Conspecific avoidance mechanisms, on the other hand, can be measured directly, e.g. \citet{arnoldetal2011}.  Therefore our model of territorial emergence has the potential to be parametrized from empirically measured interaction mechanisms.

%


\begin{acknowledgements}
This study was partly funded by NSERC Discovery and Accelerator grants (MAL, JRP).  MAL also gratefully acknowledges a Canada Research Chair and a Killam Research Fellowship.  We are grateful to Andrew Bateman and other members of the Lewis Lab for helpful discussions.
\end{acknowledgements}

\section*{Appendix A}
\label{appendix_a}

Let $u(x,t)$, $v(x,t)$, $p(x,t)$ and $q(x,t)$ be the density functions corresponding to $U(n,m)$, $V(n,m)$, $P(n,m)$ and $Q(n,m)$ respectively, where $x=an$ and $t=m\tau$.  First note the following limit as $a \rightarrow 0$, $k(a) \rightarrow \infty$, $ak(a)\rightarrow \delta$
\begin{align}
\frac{1}{2k(a)-1}\sum_{i=1-k(a)}^{k(a)-1}Q(n+i,m) = \frac{1}{2k(a)-1}\sum_{i=1-k(a)}^{k(a)-1}q(x+ai,t)a \rightarrow \frac{1}{2\delta}\int_{-\delta}^{\delta}q(x+z,t){\rm d}z.
\end{align}
Using the definition of $\bar{q}(x,t)$ given in equation (\ref{barq}), and writing equation (\ref{ume2}) down in terms of the density functions, we have
\begin{align}
&\frac{u(x,t+\tau)-u(x,t)}{\tau}=\frac{a^2}{2\tau}\biggl\{\frac{1}{a}\left[\frac{u(x+a,t)-u(x,t)}{a}-\frac{u(x,t)-u(x-a,t)}{a}\right]+\\ \nonumber
&\frac{1}{2a}\left[4u(x+a,t)\frac{\bar{q}(x+2a,t)-\bar{q}(x,t)}{2a}-4u(x-a,t)\frac{\bar{q}(x,t)-\bar{q}(x-2a,t)}{2a}\right] +\\ \nonumber
&\frac{1}{a}\biggl[\frac{u(x,t)\bar{q}(x+a,t)\bar{q}(x-a,t)-u(x-a,t)\bar{q}(x,t)\bar{q}(x-2a,t)}{a}-\\ \nonumber
&\frac{u(x+a,t)\bar{q}(x+2a,t)\bar{q}(x,t)-u(x,t)\bar{q}(x+a,t)\bar{q}(x-a,t)}{a}\biggr]\biggr\}.
\end{align}
We keep $x$ constant in the limit as $a\rightarrow 0, n\rightarrow\infty$.  Taylor expanding the right-hand side about $x$, assuming $a$ is arbitrarily small, gives the following expression
\begin{align}
\frac{u(x,t+\tau)-u(x,t)}{\tau}=\frac{a^2}{2\tau}\biggl\{&\frac{{\rm d}^2u}{{\rm d}x^2}-\left[4\frac{{\rm d}u}{{\rm d}x}\frac{{\rm d}\bar{q}}{{\rm d}x}+4u\frac{{\rm d}^2\bar{q}}{{\rm d}x^2}\right]-\nonumber \\
&\left[2u\bar{q}\frac{{\rm d}^2\bar{q}}{{\rm d}x^2}+2u\frac{{\rm d}\bar{q}}{{\rm d}x}\frac{{\rm d}\bar{q}}{{\rm d}x}+\frac{{\rm d}^2u}{{\rm d}x^2}\bar{q}^2+4\frac{{\rm d}u}{{\rm d}x}\frac{{\rm d}\bar{q}}{{\rm d}x}\bar{q}\right]+O(a)\biggr\}.
\label{uq_te}
\end{align}
In the limit as $a,\tau\rightarrow 0$ such that $a^2/(2\tau)\rightarrow D$, this simplifies to give equation (\ref{qu_time}).

\section*{Appendix B}
\label{sec:noave}

We look for solutions to equation (\ref{uv_ss_a}) in two cases: $\mbox{det}(A)\neq 0$ and $\mbox{det}(A)= 0$.  For $\mbox{det}(A)\neq 0$, one solution is to have $\dot{\bf u}={\bf 0}$, implying that $u(x)$ and $v(x)$ are constant functions so territorial patterns do not form.  

Otherwise, suppose that $\mbox{det}(A)\neq 0$, ${\rm d}u/{\rm d}x=0$ and ${\rm d}v/{\rm d}x\neq 0$.  Then the following equations hold
\begin{align}
2u(2m+v)=0, \label{uv_ss_a2} \\
(m+2u)(m+u)=0.
\label{uv_ss_a3}
\end{align}
Equation (\ref{uv_ss_a2}) implies $u=0$ or $v=-2m$.  However, if $u=0$ then equation (\ref{uv_ss_a3}) would imply $m=0$, which contradicts $\mbox{det}(A)\neq 0$.  Furthermore, if $v=-2m$ then $v<0$, which contradicts the fact that $v(x)$ is a probability density function.  In conclusion, if $\mbox{det}(A)\neq 0$, we cannot have ${\rm d}u/{\rm d}x=0$ and ${\rm d}v/{\rm d}x\neq 0$.  Similarly, if $\mbox{det}(A)\neq 0$, we cannot have ${\rm d}v/{\rm d}x=0$ and ${\rm d}u/{\rm d}x\neq 0$.  Therefore the only possible way for non-constant steady states to arise is if $\mbox{det}(A)=0$.

\begin{lemma}
\label{FiniteSolns}
If $\mbox{det}(A)=0$ then there are two possibilities.
\begin{enumerate}
\item {\bf No scent decay.} If $m=0$ then $p(x)=q(x)=1$ and $u(x)$, $v(x)$ can take any value.
\item {\bf Positive scent decay.} If $m>0$ then, for each $x\in [0,1]$, there are finitely many possible values for $u(x)$ and $v(x)$, one of which is $u(x)=v(x)=m$.  Furthermore, all solutions other than $u(x)=v(x)=m$ have $u(x),v(x)\neq m$.
\end{enumerate}
\end{lemma}
\begin{proof}
If $m=0$ then $\mbox{det}(A)=0$.  Furthermore, by equations (\ref{pu_ast}) and (\ref{qv_ast}), we have $p(x)=q(x)=1$.  Therefore ${\rm d}p/{\rm d}x={\rm d}q/{\rm d}x=0$ so that equations (\ref{qu_ss}) and (\ref{pv_ss}) hold regardless of the values of $u(x)$ and $v(x)$, proving part 1 of the lemma.  Indeed, in the time-dependent PDEs (\ref{pu_dimless},\ref{qv_dimless},\ref{qu_noave},\ref{pv_noave}), if an initial condition of $p(x,0)=q(x,0)=1$ is given and $m=0$ then ${\rm d}p/{\rm d}t={\rm d}q/{\rm d}t={\rm d}u/{\rm d}t={\rm d}v/{\rm d}t=0$ so $u(x,t)=u(x,0)$ and $v(x,t)=v(x,0)$ remain unchanged for all times $t$.

Now suppose $m\neq 0$.  For notational ease, we drop the explicit dependencies of $u$ and $v$ on $x$ for the rest of this proof, noting that they always refer to the steady states.  Then the equation $\mbox{det}(A)=0$ implies the following polynomial holds
\begin{align}
(m+2v)(m+2u)(m +u)(m+ v)=4uv(2m +u)(2m +v).
\label{det_a_poly}
\end{align}
Equation (\ref{det_a_poly}) can be rearranged to give
\begin{align}
u^2(2m-2v)+u(3m^2-7mv-2v^2)+(m^3+3m^2v+2mv^2)=0.
\label{det_a_poly2}
\end{align}
Clearly $m=u=v$ satisfies equation (\ref{det_a_poly2}). Furthermore, it follows from equation (\ref{det_a_poly2}) that $u=m$ if and only if $v=m$.  Hence any solution other than $u=v=m$ has both $u\neq m$ and $v\neq m$.

Differentiating equation (\ref{det_a_poly}) with respect to $x$, we find
\begin{align}
\frac{{\rm d}u}{{\rm d}x}=\frac{{\rm d}v}{{\rm d}x}\frac{7m u - 3m^2 - 4m v + 2u^2 + 4 uv}{3m^2 - 7m v + 4m u - 2v^2 - 4uv}.
\label{dudv_poly}
\end{align}
Using equation (\ref{dudv_poly}) and the top line of the vector equation (\ref{uv_ss_a}), together with our assumption that $du/dx,dv/dx\neq0$, we find that
\begin{align}
(7m u - 3m^2 - 4m v + 2 u^2 + 4uv)(m+2  v)(m+v)+&\nonumber \\
(4m u+2uv)(3m^2  - 7m v + 4m u - 2v^2 - 4uv)&=0.
\label{other_poly}
\end{align}

Proving part 2 of the lemma requires applying B\'ezout's Theorem \citep{fulton1969} to equations (\ref{det_a_poly2}) and (\ref{other_poly}).  B\'ezout's Theorem states that if two projective plane curves are zeros of polynomials with no non-constant greatest common divisor, then the curves intersect at finitely many points.  The polynomials on the left-hand sides of equations (\ref{det_a_poly2}) and (\ref{other_poly}) are homogeneous in three unknowns, therefore equations (\ref{det_a_poly2}) and (\ref{other_poly}) describe curves in the real projective plane.  Thus, to prove part 2 of Lemma \ref{FiniteSolns}, it suffices to show that these two polynomials have no non-constant common factor.

Let $f(m,u,v)$ be the polynomial on the left-hand side of equation (\ref{det_a_poly2}).  Since this is quadratic in $u$, it written as precisely one of the following two possible decompositions:
\begin{align}
f(m,u,v)&=[a_1(m,v)u+b_1(m,v)][a_2(m,v)u+b_2(m,v)]c(m,v), \label{det_a_ploy_fact1} 
\end{align}
or
\begin{align}
f(m,u,v)&=\alpha(m,u,v)\beta(m,v), \label{det_a_ploy_fact2}
\end{align}
where $a_1(m,v)$, $a_2(m,v)$, $b_1(m,v)$, $b_2(m,v)$, and $\beta(m,v)$ are polynomials, and $\alpha(m,u,v)$ is an irreducible polynomial.  By solving equation (\ref{det_a_poly2}) in terms of $u$, we find that
\begin{align}
u=\frac{7mv+2v^2-3m^2 \pm \sqrt{4(v-\gamma_+m)(v-\gamma_-m)(v-\vartheta_+m)(v-\vartheta_-m)}}{4(m-v)},
\label{det_a_poly3}
\end{align}
where 
\begin{align}
\gamma_{\pm}&=\frac{1}{4}\left[-11-4\sqrt{3}\pm 2\sqrt{\frac{225}{4}+30\sqrt{3}}\right], \nonumber \\
\vartheta_{\pm}&=-\frac{11}{4}+\sqrt{3}\pm \frac{1}{4}\sqrt{15(15-8\sqrt{3})}.
\label{gamma_theta}
\end{align}
Therefore the numerator of equation (\ref{det_a_poly3}) is not a polynomial, so the decomposition given in equation (\ref{det_a_ploy_fact1}) cannot hold.  

It follows that $f(m,u,v)=\alpha(m,u,v)\beta(m,v)$ where $\alpha(m,u,v)$ is irreducible and $\beta(m,v)$ is the greatest common divisor of the  coefficients of $u^n$ in $f(m,u,v)$ for $n=0,1,2$.  These coefficients are $2m-2v$, $3m^2-7mv-2v^2$ and $m^3+3m^2v+2mv^2$ (equation \ref{det_a_poly2}).  Since $m-v$ does not divide $3m^2-7mv-2v^2$ or $m^3+3m^2v+2mv^2$, it follows that $\beta(m,v)$ is a constant.  Hence $f(m,u,v)$ is irreducible.

Since $f(m,u,v)$ is both irreducible and not a constant multiple of the polynomial in equation (\ref{other_poly}), there is no non-constant greatest common divisor of the polynomials in equations (\ref{det_a_poly2}) and (\ref{other_poly}).  The proof of part 2 then follows from B\'ezout's Theorem. \qed
\end{proof}

\begin{figure*}
\includegraphics[width=\textwidth]{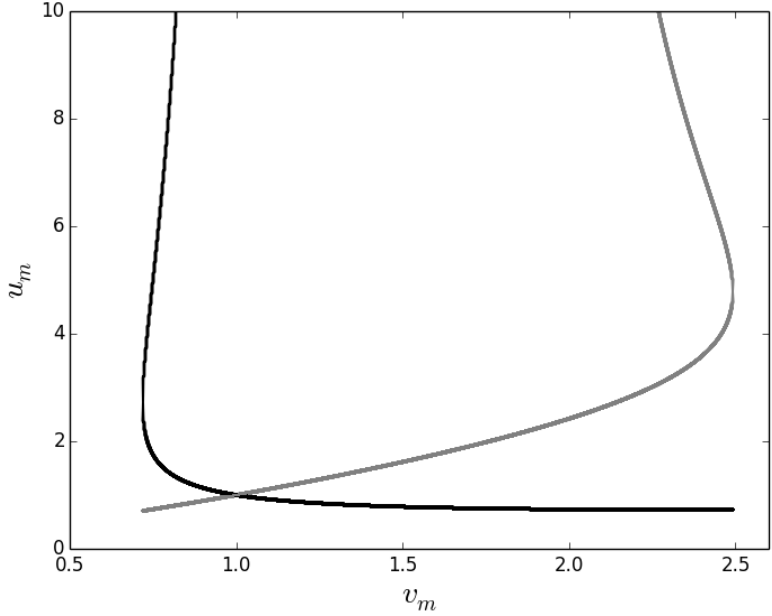}
\caption{{\bf Possible values of ${u}_m=u(x)/m$ and ${v}_m=v(x)/m$.} The black curve denotes solutions to equation (\ref{det_a_poly4}), whereas the grey curve shows solutions to equation (\ref{other_poly2}).  There appears to be only one crossing-point for positive real values of both ${u}_m$ and ${v}_m$, which is where ${u}_m={v}_m=1$, so that $u=v=m$.
}
\label{weak_soln_fig}       
\end{figure*}
Lemma \ref{FiniteSolns} enables us to prove Theorem \ref{NoTerrsLocInt} from Section \ref{sec:tp}, as follows.

\vspace{3mm}\noindent{\it Proof of Theorem \ref{NoTerrsLocInt}.} Part 1 of Theorem \ref{NoTerrsLocInt} is identical to part 1 of Lemma \ref{FiniteSolns}.  To show part 2, note that classical solutions must be continuous.  Lemma \ref{FiniteSolns} states that there are only finitely many possible values of $u$ and $v$.  Therefore any classical solution must be constant. \qed

\begin{note}
Numerical analysis suggests that $u=v=m$ is the only positive real solution (Fig. \ref{weak_soln_fig}).  Since we are interested in the case $m\neq 0$, we set ${v}_m=v/m$, ${u}_m=u/m$ and assume $v \neq m$.  Then equation (\ref{det_a_poly2}) implies
\begin{align}
{u}_m=\frac{7{v}_m+2{v}_m^2-3 \pm \sqrt{(3-7{v}_m-2{v}_m^2)^2-8(1-{v}_m)(1+3{v}_m+2{v}_m^2)}}{4(1-{v}_m)}.
\label{det_a_poly4}
\end{align}
Furthermore, equation (\ref{other_poly}) rearranges to give ${u}_m$ as another two-valued function of ${v}_m$
\begin{align}
{u}_m=&\frac{4{v}_m^3+4{v}_m^2+3{v}_m+19}{8{v}_m^2+4{v}_m-36}\pm\nonumber \\
&\frac{\sqrt{(4{v}_m^3+4{v}_m^2+3{v}_m+19)^2-4(4{v}_m^2+2{v}_m-18)(8{v}_m^3+18{v}_m^2+13{v}_m+3)}}{8{v}_m^2+4{v}_m-36}.
\label{other_poly2}
\end{align}
The black curve in Fig. \ref{weak_soln_fig} has an asymptote at ${v}_m=1$, where the denominator of the right-hand side of equation (\ref{det_a_poly4}) tends to $0$.  The grey curve has an asymptote at ${v}_m=(\sqrt{73}-1)/4$, where the denominator of the right-hand side of equation (\ref{other_poly2}) tends to $0$, so the two curves do not cross at values of ${u}_m$ higher than those shown in Fig. \ref{weak_soln_fig}.
\end{note}


\begin{thebibliography}{}
%
%
\bibitem[Adams(2001)]{adams2001}Adams ES (2001) Approaches to the study of territory size and shape. {Annu Rev Ecol Syst} 32: 277-303.

\bibitem[Arnold {\it et al.}(2011)]{arnoldetal2011}Arnold J, Soulsbury CD, Harris S (2011) Spatial and behavioral changes by red foxes ({\it Vulpes vulpes}) in response to artificial territory intrusion. Can J Zool 89:808-815

\bibitem[Bateman {\it et al.}(2015)]{batemanetal2015}Bateman AW, Lewis MA, Gall G, Manser MB, and Clutton-Brock TH (2015) Territoriality and home-range dynamics in meerkats, Suricata suricatta: a mechanistic modelling approach.  J Anim Ecol 84: 260–271

\bibitem[Briscoe {\it et al.}(2002)]{briscoeetal2002}Briscoe BK, Lewis MA \& Parrish SE (2002) Home range formation in wolves due to scent marking. {Bull Math Biol} 64: 261-284

\bibitem[Burt(1943)]{burt1943}Burt WH (1943) Territoriality and home range concepts as applied to mammals. {J Mammal} 24: 346-352.

\bibitem[Durrett \& Levin(1994)]{durrettlevin1994}Durrett R \& Levin S (1994) The importance of being discrete (and spatial). {\it Theor. Pop. Biol.}, 46, 363-394.

\bibitem[Einstein(1916)]{einstein1916}Einstein A (1916) The Foundation of the General Theory of Relativity. Annalen der Physik. 354: 769-822

\bibitem[Fulton(1969)]{fulton1969}Fulton W (1969) Algebraic Curves. Mathematics Lecture Note Series, W.A. Benjamin, New York.

\bibitem[Giuggioli \& Kenkre(2014)]{giuggiolikenkre2014}Giuggioli L, Kenkre VM (2014) Consequences of animal interactions on their dynamics: emergence of home ranges and territoriality.  Move. Ecol. 2:20 doi:10.1186/s40462-014-0020-7

\bibitem[Giuggioli {\it et al.}(2011a)]{GPH1}Giuggioli L, Potts JR, Harris S (2011a) Animal interactions and the emergence of territoriality. {PLoS Comput Biol}, 7:1002008

\bibitem[Giuggioli {\it et al.}(2011b)]{GPH2}Giuggioli L, Potts JR, Harris S (2011b) Brownian walkers within subdiffusing territorial boundaries. {Phys Rev E}, 83:061138

\bibitem[Harris(1980)]{harris1980}Harris S (1980) Home ranges and patterns of distribution of foxes ({\it Vulpes vulpes}) in an urban area, as revealed by radio tracking. In: Amlaner CJ  \& Macdonald DW (eds) Handbook of biotelemetry and radio tracking, Pergamon Press, Oxford, pp 685-690

\bibitem[Lewis {\it et al.}(1997)]{lewisetal1997} Lewis MA, White KAJ \& Moorcroft PR (1997) Analysis of a model for wolf territories. {J Math Biol} 35: 749-774.

\bibitem[Lewis \& Murray(1993)]{lewismurray1993}Lewis MA, Murray JD (1993) Modelling territoriality and wolf-deer interactions. {Nature} 366:738-740.

\bibitem[Mallinson \& de Vahl Davis(1973)]{mallinsondevahldavis1973}Mallinson GD, de Vahl Davis G (1973) The method of the false transient for the solution of coupled elliptic equations. J Comp Phys 12:435-461.

\bibitem[Moorcroft \& Lewis(2006)]{moorcroftlewis2006}Moorcroft PR \& Lewis MA (2006) {Mechanistic Home Range Analysis}. Princeton University Press, Princeton.

\bibitem[Moorcroft {\it et al.}(2006)]{moorcroftetal2006}Moorcroft PR, Lewis MA \& Crabtree RL (2006) Mechanistic home range models capture spatial patterns and dynamics of coyote territories in Yellowstone. {Proc Roy Soc B} 273: 1651-1659

\bibitem[Moorcroft(2012)]{moorcroft2012}Moorcroft PR (2012) Mechanistic approaches to understanding and predicting mammalian space use: recent advances, future directions.  {J Mammal} 93: 903-916.

\bibitem[Murray(2002)]{murray2002}Murray JD (2002) Mathematical biology II: spatial models and biomedical applications. 3rd ed.  Springer-Verlag, New York.

\bibitem[Potts {\it et al.}(2012)]{PHG2}Potts JR, Harris S, Giuggioli L (2012) Territorial dynamics and stable home range formation for central place foragers. {PLoS One}, 7:0034033

\bibitem[Potts \& Lewis(2014)]{pottslewis2014}Potts JR, Lewis MA. (2014) How do animal territories form and change? Lessons from 20 years of mechanistic modelling. Proc Roy Soc B 281:20140231

\bibitem[Potts {\it et al.}(2014)]{PML2013}Potts JR, Mokross K, Lewis MA (2014) A unifying framework for quantifying the nature of animal interactions. J Roy Soc Interface 11:20140333

 
    
\end{thebibliography}
\end{document}